\documentclass[a4paper,11pt]{article}% \usepackage[utf8]{inputenc}
\usepackage{amsthm,amssymb,amsmath}
\usepackage{xcolor}
\usepackage{multirow}
\usepackage{algorithm}
\usepackage[noend]{algpseudocode}
\usepackage[hidelinks]{hyperref}
\newtheorem{theorem}{Theorem}[section]
\newtheorem{lemma}[theorem]{Lemma}
\newtheorem{proposition}[theorem]{Proposition}

\newtheorem{corollary}[theorem]{Corollary}

\author{Gwenaël Richomme\footnote{Université Paul-Valéry Montpellier 3, Université de Montpellier, CNRS, Montpellier, France},
Matthieu Rosenfeld\footnote{Université de Montpellier, CNRS, Montpellier, France}}
\usepackage{geometry}
\title{Reconstructing words using queries on subwords or factors}

\newcommand{\oq}{``}%opening quote
\newcommand{\cq}{''}%closing quote
\begin{document}

\maketitle

\begin{abstract}
We study word reconstruction problems.
Improving a previous result by P.~Fleischmann, M.~Lejeune, F.~Manea, D.~Nowotka and M.~Rigo,
we prove that,
for any unknown word $w$
of length $n$ over an alphabet of cardinality $k$,
$w$ can be reconstructed
from the number of occurrences as subwords (or scattered factors)
of $O(k^2\sqrt{n\log_2(n)})$ words.
Two previous upper bounds obtained by S.~S.~Skiena and G.~Sundaram
are also slightly improved:
one when considering information on the existence of subwords instead of on the numbers of their occurrences,
and,  the other when considering information on the existence of factors.
\end{abstract}
\pagebreak

\section{Introduction}
A natural combinatorial question is to ask how much partial information on an object is needed to reconstruct this object (see below and in our references for examples).
For example,
in \cite{Fleischmann_et_al2020DLT,Fleischmann_et_al2021IJFCS},
P.~Fleischmann, M.~Lejeune, F.~Manea, D.~Nowotka and M.~Rigo
consider the problem of reconstructing a word $w$ from information on the number of occurrences as subwords of $w$ of some words.
Let us recall that a word $u$ is a \textit{subword} of a word $w$ (or a \textit{scattered subword} of $w$) if
$u$ and $w$ can be decomposed in the form $u = u_1 \cdots u_\ell$ and $w = v_0 u_1 v_1 \cdots u_\ell v_\ell$
for some words $u_1, \ldots, u_\ell, v_0, \ldots, v_\ell$.
Such a double decomposition marks an occurrence of $u$ as a subword of $w$.
The number of occurrences of $u$ as a subword of $w$ is sometimes denoted as the binomial coefficient
$\binom{w}{u}$ since this number coincides with the traditional coefficient
$\binom{|w|}{|u|}$ when the words $u$ and $w$ are written on a single letter
(here, as usual in combinatorics on words, $|w|$ denotes the length of $w$), see for instance \cite[chap. 6]{Lothaire1983book}.
The problem addressed by Fleischmann \textit{et al.} is presented as a game in which the player has to guess an unknown word.
In his task the player asks questions in a certain form until he has enough information to uniquely determine the word.
More precisely, at each round, the player chooses a word $u$ based on the previous answers that he obtained and asks for the value of $\binom{w}{u}$. The goal of the player is to minimize the number of questions.
Fleischmann \textit{et al.} proved that there is a strategy to ensure that at most $\min(|w|_a, |w|_b)+1\le \lfloor\frac{|w|}{2}\rfloor+1$ questions are needed when $w$ is defined on the binary alphabet $\{a, b\}$ (for a letter $\alpha$, $|w|_\alpha = \binom{w}{\alpha}$ denotes the number of occurrences of $\alpha$ in $w$).
For any word $w$ over the alphabet $\{1, \ldots, k\}$
they proved that the number of questions needed is bounded  by $\sum_{i \in \{1, \ldots, k\}} |w|_i (k+1-i)$.
Our main results (Theorem~\ref{T_binarycountsubword} and Corollary~\ref{C_sharp_arbitrary_alphabet}) prove that this number of questions is at most  $\binom{k}{2}\left(7\left\lceil \sqrt{|w| \log_2 (|w|)}\right\rceil+4\right)$. For any fixed $k$, our upper bound is asymptotically much stronger as the length of the word goes to infinity. For binary words in particular, their upper bound is $\frac{|w|}{2}+1$ and ours is $7\left\lceil \sqrt{|w| \log_2 (|w|)}\right\rceil+4\,.$ We also adapt this strategy (Theorem \ref{expectrunningtime}) to provide an algorithm whose expected running time over a uniform random binary word of length $n$ is $\mathcal{O}(\log_2 n)$.

Let us recall that the previous game is related to another problem that seems to have been first introduced by L.~O.~Kalashnik \cite{Kalashnik1973}:
What is the smallest $\ell$ such that we can reconstruct $w$ from the values $\binom{w}{u}$ for all words $u$ of length $\ell$?
As far as we know, the best upper bound, $\lfloor \frac{16}{7}\sqrt{|w|}\rfloor +5$, for this problem was obtained by I.~Krasikov and Y.~Roditty in 1997 \cite{KrasikovRoditty1997JCTA} using a link with the Prouhet-Tarry-Escott problem about Diophantine analysis.
Also the best known lower bound, $3^{(\sqrt{2/3}-o(1))\log^{1/2}_3(|w|))}$, is due to \cite{Dudik_Schulman2003JCTA}.
Our result does not improve this upper bound since, in the binary case,
at least one query concerns a word $u$ of length at least $\min(|w|_0, |w|_1)$ which is around $|w|/2$ for many words $w$.

In a variant of the previous problem queries in the form {\oq}what is the value of $\binom{w}{u}$?{\cq}
is replaced with queries in the form {\oq}Is $\binom{w}{u} \geq 1$?{\cq}
or equivalently {\oq}Is $u$ a subword of $w$?{\cq}.
More precisely the problem is to determine the least value $\ell$
such that the set of subwords of length $\ell$ determines uniquely a word $w$.
This problem arose in various areas.
In \cite[Chap 6]{Lothaire1983book}, it is proved
that any word $w$ of length $n$ over an alphabet ${\cal A}$
is uniquely determined by its set of
subwords in the form $a^*b^*$ of length at most $\lceil |w|_a + |w|_b+1/2 \rceil$ with
$a$ and $b$ distinct letters of ${\cal A}$.
The problem is also studied in \cite{Levenshtein2001JCTA}.

In \cite{Skiena_Sundaram1993LNCS,Skiena_Sundaram1995JCB},
in the context of DNA sequencing of hybridization,
S.~S.~Skiena and G.~Sundaram consider the problem
of minimizing the number of queries in the form {\oq}Is $u$ a subword of $w$?{\cq}.
They prove that a word $w$ of length $n$ over an alphabet ${\cal A}$
of cardinality $k$ can be reconstructed
using $O(n \log_2(k) +k \log_2(n))$ such queries.
More precisely Theorem~15 in \cite{Skiena_Sundaram1995JCB}
states that $ 1.59 n \log_2(k) + 2k \log_2(n) + 5k$ queries are sufficient to reconstruct $w$.
Using a basic information theory approach
S.~S.~Skiena and G.~Sundaram also provide the lower bound $n \log_2 k$ for the number of queries.
In Section~\ref{Sec:exist_queries},
we slightly improve S.~S.~Skiena and G.~Sundaram's strategy
and we provide a new upper bound, reducing the gap with the lower bound.
More precisely,
we state that at most $n \log_2(k) + k (2 + \lfloor \log_2(n+1)\rfloor)$
queries are sufficient to reconstruct $w$, reducing the gap between the bounds from $0.59n \log_2(k) +O(k \log_2(n))$ down to $O(k \log_2(n))$.

In Section~\ref{Sec:exists-factor_bis}, we consider factors instead of subwords (a word $u$ is a \textit{factor} of a word $w$ if there exist words $p$ and $s$ such that $w = pus$) and  the corresponding problem of minimizing the number of queries in the form {\oq}Is $u$ a factor of $w$?{\cq} needed to reconstruct an unknown word $w$.
In \cite{Skiena_Sundaram1993LNCS,Skiena_Sundaram1995JCB},
S.~S.~Skiena and G.~Sundaram prove that,  for an unknown word $w$ over an alphabet ${\cal A}$
of cardinality $k$,
if the length $n$ of $w$  is known then $w$ can be reconstructed using a number of queries which is in
$(k-1)n+ 2\log_2(n)+O(k)$.  
Actually their proof leads to the upper bound $(k-1)n+\log_2(n)+O(k)$, which is $n+\log_2(n)+O(1)$ in the binary case. This more accurate upper bound was already mentioned in the binary case in \cite{Skiena_Sundaram1995JCB}.
A simple double counting argument (there are $k^n$ words of length $n$ and each question has two possible outcomes) leads to the lower bound $n\log_2 k$.
We improve their strategy and reduce the upper bound to $(k-1)(n+2) + \left\lceil \frac{\log_2(n)}{2}\right\rceil+3$. In the binary case, this reduces the gap between the lower and the upper bound from $\log_2(n)+O(1)$ down to $\left\lceil\frac{\log_2(n)}{2}\right\rceil+5$.

Queries in the form {\oq}What is the number of occurrences of a word $u$ as a factor of $w${\cq} have also been considered by S.S. Skiena et G. Subraman \cite{Skiena_Sundaram1995JCB}. Their lower bound $nk/4- o(n)$ on the number of queries needed is, up to our knowledge, the best known. One can deduce whether a word $u$ occurs as a factor in a word $w$ from the number of occurrences of $u$ in $w$. This observation allows them to obtain the same upper bounds for this fourth problem than for the previous problem. Similarly, our bound applies. Hence, we also slightly improve the upper bound in this case, but this improvement is negligible compared to the size of the gap between the lower bound and the upper bound.

Basic definitions and notations have already been recalled (following \cite{Lothaire1983book}).
Let us observe that $\#S$ denotes the cardinality of a set $S$.
Moreover, given a word $w$ over an alphabet ${\cal A}$,
we will simply use $n$ to denote the length $|w|$ of $w$ and
$k$ to denote the cardinality $\#{\cal A}$ of ${\cal A}$.

%As shown by the introduction, considered problems deal with the reconstruction or determination of a word from information provided by answers to queries on factors or subwords. As solutions are thought as strategies as if the problems were games, we sometimes use here the term “guess” as a synonym to “determine”.

\section{How-many-subwords queries\label{Sec:How many}}
In this section, we focus on queries in the form
{\oq}How many occurrences of $u$ as a subword does $w$ contains?{\cq}
or equivalently
{\oq}What is the value of $\binom{w}{u}$?{\cq}.
We call such a query a $\#$-subword query.
Our main result regarding this kind of query is the following. Of course, as it will be the case for other queries in the next sections, we assume that such a query can be answered without knowing $w$.
\begin{theorem}\label{T_binarycountsubword}
The number of $\#$-subword queries needed to reconstruct a word of length $n$ over $\{0,1\}$ is at most $7\left\lceil \sqrt{n \log n}\right\rceil+4$ whether $n$ is known or not.
\end{theorem}

A word $w$ that contains $m$ occurrences of $1$, can always be written as
$w=0^{s_0}10^{s_1}1\ldots 10^{s_{m}}$ where the $s_i$ are nonnegative integers.
Since $m = \binom{w}{1}$, it only requires one query to find $m$.
Our goal is to find the values of all the $s_i$.
Our strategy relies on the fact that if we know which of the $s_i$ are {\oq}large{\cq} and if we know their values then we can determine multiple others $s_i$ with a single query (this is shown in Lemma \ref{coeffromeq}). On the other hand since we cannot have too many {\oq}large{\cq} $s_i$ we have an efficient strategy to find all these $s_i$ (see Lemma \ref{locatelargblocks}).
Using these two facts together and optimizing the meaning of {\oq}large{\cq} we get the desired result.

Actually, in a uniform random word we do not expect to have any $s_i$ larger than $\mathcal{O}(\log n)$ and this leads to a more efficient average case algorithm.
\begin{theorem}\label{expectrunningtime}
There is a deterministic strategy that, given any integer $n$, reconstructs in average in  $\mathcal{O}(\log_2(n))$ queries any word $w$ taken uniformly at random among all binary words of length $n$.
\end{theorem}

The next lemma allows to prove Lemma~\ref{coeffromeq}.

\begin{lemma}\label{coeffromeqanoying}
  Let $r$, $\ell$, $s_1,\ldots,s_r$ be non-negative integers such that $1 \leq r \leq \ell +1$ and
  for all $j \in \{1, \ldots, r\}$, $s_j < \frac{\ell+1}{r}$.
  The values of $s_1,\ldots,s_r$ are uniquely determined by the values of
  $\binom{0^{s_r}10^{s_{r-1}}1\cdots 0^{s_2}10^{s_1}1^\ell}{01^\ell}$, $r$ and $\ell$.
\end{lemma}

\begin{proof}
Let us first express the number of occurrences of $01^\ell$ as subword in $0^{s_r}10^{s_{r-1}}1\cdots 0^{s_1}1^\ell$. By considering separately the different possible positions of the $0$ in the occurrence we obtain
\begin{equation}\label{eq_initial_observation}
  \binom{0^{s_r}10^{s_{r-1}}1\cdots 0^{s_2}10^{s_1}1^\ell}{01^\ell}= \sum_{j=1}^r s_j\binom{\ell+j-1}{\ell}=\sum_{j=1}^r s_j\binom{\ell+j-1}{j-1}\,.
\end{equation}

  Let $\beta=\max_{j}s_j$.
 We first show that for all $t\in\{1,\ldots, r\}$,
 \begin{equation}\label{eqindbinomcoef}
\sum_{j=1}^t s_j\binom{\ell+j-1}{j-1}\le \beta \binom{\ell+t}{t-1}\,.
 \end{equation}
  We proceed by induction on $t$. It is easily verified for $t=1$. Now if \eqref{eqindbinomcoef} holds for $t$, then
 \begin{align*}
\sum_{j=1}^{t+1} s_j\binom{\ell+j-1}{j-1}&=\sum_{j=1}^{t} s_j\binom{\ell+j-1}{j-1}+s_{t+1}\binom{\ell+t}{t}\le \beta \binom{\ell+t}{t-1}+s_{t+1}\binom{\ell+t}{t}\\
&\le \beta \left(\binom{\ell+t}{t-1}+\binom{\ell+t}{t}\right) =\beta \binom{\ell+t+1}{t}
 \end{align*}
  which concludes the inductive proof of \eqref{eqindbinomcoef}.
 
  Moreover, for all $t\in\{1,\ldots, r\}$,
  $\beta \binom{\ell+t}{t-1}< \frac{\ell+1}{r}\binom{\ell+t}{t-1}\le  \frac{\ell+1}{t}\binom{\ell+t}{t-1}= \binom{\ell+t}{t}$. Together with \eqref{eqindbinomcoef}, it implies that for all  $t\in\{1,\ldots, r\}$,
 \begin{equation}\label{eqnegligible}
0\le\sum_{j=1}^t s_j\binom{\ell+j-1}{j-1}< \binom{\ell+t}{t}\,.
 \end{equation}
Observe that, for all $t\in\{1,\ldots, r-1\}$, 
\begin{equation*}
    s_{t+1}=\frac{\sum_{j=1}^{t+1} s_j\binom{\ell+j-1}{j-1}-\sum_{j=1}^t s_j\binom{\ell+j-1}{j-1}}{\binom{\ell+t}{t}}\,.
\end{equation*}
But $s_{t+1}$ is an integer and by equation\eqref{eqnegligible} the right part of the fraction in the left-hand-side is in $[0,1[$ we deduce
  \begin{equation}\label{eq_deduced_value}
s_{t+1}=  \left\lfloor\frac{\sum_{j=1}^{t+1} s_j\binom{\ell+j-1}{j-1}}{\binom{\ell+t}{t}} \right\rfloor\,.
  \end{equation}
    By~Equations~\eqref{eq_initial_observation} and \eqref{eq_deduced_value}, we can deduce the value of $s_r$ from $r$, $l$ and
    $\sum_{j=1}^{r} s_j\binom{\ell+j-1}{j-1}$ which is itself deduced from $\binom{0^{s_r}10^{s_{r-1}}1\cdots 0^{s_2}10^{s_1}1^\ell}{01^\ell}$.
    From the value of $s_r$, we can now deduce $\sum_{j=1}^{r-1} s_j\binom{\ell+j-1}{j-1}$ and thus $s_{r-1}$ by  \eqref{eq_deduced_value}.
    Thus, by an {\oq}inverse induction{\cq}  from $r-1$ to $1$, we deduce the values of all the $s_j$.
\end{proof}

Lemma \ref{coeffromeqanoying} allows us to determine the length of multiple consecutive $0$-blocks with only one query under some strong hypothesis, but we can relax these hypotheses as follows. The idea is that if we have some large $s_i$ and a prefix, it is enough to know the value of these $s_i$ and of the prefix in order to remove their contribution before applying the previous lemma.
\begin{lemma}\label{coeffromeq}
Let $p$ and $v$ be words, $r$ and $s_1,\ldots, s_r$ be nonnegative integers such that $1\le r\le|v|_1 + 2$ and let $w=p0^{s_r}10^{s_{r-1}}\ldots 10^{s_1}1v$.
Suppose that $p$, $|v|_1$ and $r$ are known and that for all $j$, either $s_j$ is known or $s_j<\frac{|v|_1+2}{r}$, then the value of $\binom{w}{01^{1+|v|_1}}$ uniquely determines the values of all the unknown $s_j$ for $j\in\{1,\ldots,r\}$.
\end{lemma}
\begin{proof}
For all $j\in\{1,\ldots,r\}$, let $s'_j$ be such that if $s_j<\frac{|v|_1+2}{r}$, then $s'_j=s_j$ and $s'_j=0$ otherwise. Then $s_j-s'_j$ is known for all $j$ (it is $s_j$ if $s_j$ is known and $0$ otherwise) and for all $j$, $s'_j<\frac{|v|_1+2}{r}$.

Now, by considering the possible positions of the $0$ in the occurrences of $01^{1+|v|_1}$, we get
\begin{align*}
\binom{w}{01^{1+|v|_1}}
&= \binom{p1^{r+|v|_1}}{01^{1+|v|_1}}+\binom{0^{s_r}10^{s_{r-1}}\ldots 10^{s_1}1^{1+|v|_1}}{01^{1+|v|_1}}\\
&=\binom{p1^{r+|v|_1}}{01^{1+|v|_1}}+ \sum_{j=1}^r s_j \binom{j+|v|_1}{1+|v|_1}\\
&=\binom{p1^{r+|v|_1}}{01^{1+|v|_1}}+ \sum_{j=1}^r (s_j-s'_j) \binom{j+|v|_1}{1+|v|_1}+\sum_{j=1}^r s'_j\binom{j+|v|_1}{1+|v|_1}\\
&=\binom{p1^{r+|v|_1}}{01^{1+|v|_1}}+ \sum_{j=1}^r (s_j-s'_j) \binom{j+|v|_1}{1+|v|_1}+\binom{0^{s'_r}10^{s'_{r-1}}\ldots 10^{s'_1}1^{1+|v|_1}}{01^{1+|v|_1}}\,.
\end{align*}
It implies that,
$$\binom{0^{s'_r}10^{s'_{r-1}}\ldots 10^{s'_1}1^{1+|v|_1}}{01^{1+|v|_1}}=\binom{w}{01^{1+|v|_1}}-\binom{p1^{r+|v|_1}}{01^{1+|v|_1}}- \sum_{j=1}^r (s_j-s'_j) \binom{j+|v|_1}{1+|v|_1}\,.$$
By assumption, $\binom{w}{01^{1+|v|_1}}$, $p$, $r$, $|v|_1$ and for all $j$, $(s_j-s'_j)$ are known. Hence, the quantity $\binom{0^{s'_r}10^{s'_{r-1}}\ldots 10^{s'_1}1^{1+|v|_1}}{01^{1+|v|_1}}$ is uniquely determined. For all $j$, $s'_j<\frac{|v|_1+2}{r}$ and we deduce from Lemma \ref{coeffromeqanoying} that the values of all the $s'_j$ are uniquely determined which concludes our proof.
\end{proof}

For any word $w$ over $\{0, 1\}$ decomposed as $w=0^{s_0}10^{s_1}1\cdots0^{s_{t-1}}10^{s_{t}}$, we call $i$ the \emph{index} of the $0$-block $0^{s_{i}}$.
If we want to use the previous lemma to reconstruct a word, we first need to determine the indices of all the $0$-blocks that are longer than some predetermined length.

\begin{lemma}\label{locatelargblocks}
Let $w\in\{0,1\}^*$ and $m$ be an integer.
Let $I$ be the set of indices of $0$-blocks of $w$ of length at least $m$. Suppose that we know $|w|$ and $|w|_0$ (and so also $|w|_1 = |w|-|w|_0$), then the number of \#-subword queries needed to determine $I$ is at most
$$\frac{2 |w|_0\lceil\log_2 (|w|_1+1)\rceil}{m}\,.$$
\end{lemma}

\begin{proof}
We use Algorithm \ref{alglongblocks} to determine $I$ calling it with $\ell = 0$ and $u = |w|_1$.
Note that $|w|_1 = |w| - |w|_0$ is known.
\begin{algorithm}[ht]
\caption{An algorithm that prints the indices $i\in\{\ell,\ldots,u\}$ of the $0$-blocks\label{alglongblocks} of length at least $m$ that occur in $w$}
\begin{algorithmic}
\Procedure{Recblocks}{$w$, $m$, $\ell$, $u$}
    \If{$\binom{w}{1^\ell0^m 1^{|w|_1-u}}\ge1$}  
        \If{$u=\ell$}
            \State \textbf{Print} $\ell$ 
        \Else
            \State\Call{Recblocks}{$w$, $m$, $\ell$, $\lfloor\frac{\ell+u}{2}\rfloor$}
            \State\Call{Recblocks}{$w$, $m$, $\lfloor\frac{\ell+u}{2}\rfloor+1$, $u$}
        \EndIf
    \EndIf
\EndProcedure
\end{algorithmic}
\end{algorithm}

The condition of the main {\oq}if{\cq} verifies that the lengths of the $0$-blocks whose indices are in $\{\ell,\ldots,u\}$ sum to at least $m$. If it doesn't then we know that none of these blocks can have length at least $m$ so we do not need to call the function recursively on any of them. From this, verifying the correctness of the algorithm is rather straightforward.

Let us now bound the total number of queries.
For this, we consider the \emph{tree of recursive calls to Recblocks} defined as follows: the root of the tree is the initial call with $\ell = 0$ and $u = |w|_1$; a call $a$ is the child of another call $b$ if the call $a$ was made in $b$. The \emph{depth} of a call is its distance to the root. The \emph{weight} of a call is the quantity $u+1-\ell$. For any call of weight $x$,
the weights of its children are $\lceil x/2\rceil$ or $\lfloor x/2\rfloor$ (and the sum of the weights of the two children is $x$).
Let $f$ be the function such that $f:x\rightarrow\lceil\frac{x}{2}\rceil$. 
The root has weight $|w|_1+1$ and $f$ is a non-decreasing function, so any call of depth $d$ has weight at most $f^d(|w|_1+1)$.
For any integer $x$, $f(x) \le \frac{x+1}{2}$, and, in particular, for all $d\ge1$,
$f^{d}(|w|_1+1)\le \frac{f^{d-1}(|w|_1+1)+1}{2}$. 
By induction on $d$, $f^d(|w|_1+1)< \frac{|w|_1+1}{2^d}+1$.
Any call of depth $\lceil\log_2 (|w|_1+1)\rceil$ has weight at most $1$ (the weight is an integer smaller than $2$) and is a leaf of the tree. Hence, the depth of any call is at most $\lceil\log_2 (|w|_1+1)\rceil$.

Moreover, one easily verifies by induction on the depth that for any two different calls $c$ and $c'$ at the same depth the corresponding intervals $[\ell,u]$ and $[\ell',u']$ are disjoint. We say that a call with the values $\ell$ and $u$ \emph{owns} the occurrences of $0$ that belongs to all the blocks of indices between $\ell$ and $u$. Then by the previous remark, the set of occurrences of $0$ owned by two calls at the same depth are disjoint.
Since the condition of the first {\oq}if{\cq} is true if the call owns at least $m$ occurrences of $0$, we deduce that there are at most $\frac{|w|_0}{m}$ such calls on any given depth.
Since each such call has two children, we deduce that the number of calls at any depth is at most $2\frac{|w|_0}{m}$. Hence the total number of calls, is at most $\frac{2 |w|_0\lceil\log_2 (|w|_1+1)\rceil}{m}$. Since we ask one query by call this concludes the proof.
\end{proof}

We are now ready to show our main result.
We will first use the algorithm from Lemma \ref{locatelargblocks} to find all the blocks
that are of length $\left\lceil \sqrt{n \log n}\right\rceil$
and then we use Lemma \ref{coeffromeq} to determine all the other blocks.

\begin{proof}[Proof of Theorem \ref{T_binarycountsubword}]~

\textbf{Phase 1}.
Let $w$ be the unknown word. 
It costs two queries to get $|w|_0=\binom{w}{0}$ and $|w|_1=\binom{w}{1}$.
Then $n = |w| = |w|_0+|w|_1$ is known. 
Suppose without loss of generality that $\binom{w}{0}\ge n/2 \ge\binom{w}{1}$ (otherwise simply exchange the role of $0$ and $1$ in the following).

\textbf{Phase 2}.
Let $m=\left\lceil \sqrt{n \log n}\right\rceil$.
We use the algorithm from Lemma \ref{locatelargblocks} to locate all the $0$-blocks of length at least $m$. There are at most $\frac{n}{m}$ such blocks and we can use one query for each of them to determine their respective length:
Indeed if the block is at index $i$ with $i \in \{0, \ldots, |w|_1\}$,
its length is $\binom{w}{1^i01^{|w|_1-i}}$.
Thus locating $0$-blocks of length at least $m$ together with their lengths
require at most $\frac{2 |w|_0\lceil\log (|w|_1+1)\rceil}{m} + \frac{n}{m} $  queries.
This number of queries is less than $3\frac{n \log n}{m}\le3 \sqrt{n \log n}$.

\textbf{Phase 3}.
We now need to determine the lengths of $0$-blocks of length at most $m$.
We first determine the $0$-blocks occurring before the $\left\lceil \frac{|w|_1}{2} \right\rceil$ last occurrences of $1$.
Secondly, we determine the $0$-blocks occurring after the $\left\lceil \frac{|w|_1}{2} \right\rceil$ first occurrences of $1$.
After this, the lengths of all the $0$-blocks are known and we know $w$.
We describe only how to determine the first half of the blocks, since reconstructing the second half of the blocks can be done symmetrically.

There are $\left\lceil \frac{|w|_1}{2} \right\rceil+1$ $0$-blocks before the $\left\lceil \frac{|w|_1}{2} \right\rceil$ last occurrences of $1$.
We determine the unknown blocks among them in at most $m$
steps from left to right considering, at each step, at most $r = \left\lfloor\frac{|w|_1}{2m}\right\rfloor$ blocks. Since $mr\ge\frac{|w|_1}{2}-m$, we might miss up to $m+1$ blocks after this, that we can recover one by one for up to $m+1$ extra queries.
At one step $w = p0^{s_r}10^{s_{r-1}}\cdots 10^{s_1}1v$
with $p$ an already known prefix of $w$ (initially $p$ is the empty word) and $|v|_1 \geq \left\lceil \frac{|w|_1}{2} \right\rceil$.
For each $i \in \{1, \ldots, r\}$,
if $s_i$ is unknown then $s_i < m = \frac{|w|_1/2}{|w|_1/(2m)} < \frac{|v|_1+2}{r}$.
By Lemma~\ref{coeffromeq},
only one query is needed to know the $r$ blocks.
Hence, we determine the $0$-blocks occurring before the $\left\lceil \frac{|w|_1}{2} \right\rceil$ last occurrences of $1$
in at most $2m+1 = 1 + 2\left\lceil \sqrt{n \log n} \right\rceil$ queries (and similarly to know the $0$-blocks
occurring after the $\left\lceil \frac{|w|_1}{2} \right\rceil$ last occurrences of $1$).

In total, our strategy uses
$2+3\left\lceil \sqrt{n \log n}\right\rceil+2(1 + 2\left\lceil \sqrt{n \log n} \right\rceil) = 7 \left\lceil \sqrt{n \log n} \right\rceil+4$.
\end{proof}
For any alphabets $\mathcal{A}$ and $\mathcal{B} \subseteq \mathcal{A}$ and any word $u$ over $\mathcal{A}$,
the \textit{projection} of $u$ onto $\mathcal{B}$
is the word obtained by removing from $u$ any letter that does not belong to $\mathcal{B}$.
We denote it $\pi_{\mathcal{B}}(u)$. For instance, $\pi_{\{0,1\}}(0120201)=01001$.
Over an alphabet of cardinality $k$ if we know the projections over all the binary sub-alphabets, we can uniquely determine the whole word \cite[Lemma 6.2.19]{Lothaire1983book}. So Theorem \ref{T_binarycountsubword} has the following corollary.
\begin{corollary}\label{C_sharp_arbitrary_alphabet}
The number of $\#$-subword queries needed to reconstruct a word of length $n$ over an alphabet of cardinality $k$ is at most $\binom{k}{2}(7\left\lceil \sqrt{n \log n}\right\rceil+4)\,.$
\end{corollary}
In Theorem \ref{T_binarycountsubword} and Corollary \ref{C_sharp_arbitrary_alphabet}, we did not try to optimize the multiplicative constant, because we believe that the $\sqrt{n \log n}$ bound is not {\oq}sharp up to a multiplicative constant{\cq}. As suggested by Theorem \ref{expectrunningtime}, the number of required queries in Theorem 1 and Corollary 6 might be in $O(\log n)$.

As we will see in Lemma~\ref{expectnolongbloc}, the probability that there is a $0$-block of length more than $\lceil2\log_2(n)\rceil$ is small. 
\begin{lemma}\label{expectnolongbloc}
Let $w$ be a word taken uniformly at random among all binary words of length $n$. The probability that
$w$ contains the factor $0^{\lceil2\log_2(n)\rceil}$ is at most $1/n$.
\end{lemma}
\begin{proof}
Let $m=\lceil2\log_2(n)\rceil$. Let $w_1,\ldots, w_n\in\{0,1\}$ be such that $w=w_1\cdots w_n$.
For all $i\in\{1,\ldots, n-m+1\}$, let $E_i$ be the event that $w_iw_{i+1}\ldots w_{i+m-1}=0^m$. Then for all $i$, $\mathbb{P}(E_i)=2^{-m}\le 1/n^2$. By union bound,
$$\mathbb{P}(0^m \text{ is a factor of } w)= \mathbb{P}(\cup_{i=1}^{n-m+1} E_i)\le\sum_{i=1}^{n-m+1}\mathbb{P}(E_i)\le \frac{1}{n}$$ as desired.
\end{proof}
\begin{proof}[Proof of Theorem \ref{expectrunningtime}]
First, we determine the number of $0$ and $1$ in $w$ in 2 queries. 
Let $m=\lceil2\log_2(n)\rceil$.
We first assume that there is no factor $0^m$ in $w$. 
We can now apply Lemma~\ref{coeffromeq} as in Phase 3 of the proof of Theorem \ref{T_binarycountsubword}, but with $m=\lceil2\log_2(n)\rceil$. We now have a candidate word $w'$ and we can ask one more question, $\binom{w}{w'}$, to verify if $w = w'$ (this might not be the case, if our starting assumption was false). All of this take $\mathcal{O}(\log_2(n))$ queries. 

If we did not obtain the correct word, we know that our assumption was false and we use Theorem \ref{T_binarycountsubword} to find $w$ in $\mathcal{O}(\sqrt{n\log_2(n)})$ extra queries.  By Lemma \ref{expectnolongbloc}, this happens with probability at most $1/n$, so the expected number of queries of this procedure is at most $\mathcal{O}(\log_2(n))+\mathcal{O}(\sqrt{n\log_2(n)}/n)=\mathcal{O}(\log_2(n))$.
\end{proof}
\section{Exists-subword queries\label{Sec:exist_queries}}
In this section, we focus on queries in the form {\oq}Is $u$ a subword of $w$?{\cq}
or equivalently
{\oq}Is $\binom{w}{u} \geq 1$?{\cq}.
We call such a query an $\exists$-subword query.
The reconstruction problem using $\exists$-subword queries of a word $w$ of unknown length $n$
over an alphabet $\mathcal{A}$ of cardinality $k$ was solved by S.~S.~Skiena and G.~Sundaram \cite{Skiena_Sundaram1993LNCS,Skiena_Sundaram1995JCB} using
$1.59 n \log_2(k) + 2k \log_2(n) + 5k$  queries.
We improve the main coefficient of the bound, replacing $1.59$ by $1$ which is optimal (any such algorithm requires at least  $n\log_2(k)$ queries in the worst case \cite{Skiena_Sundaram1993LNCS,Skiena_Sundaram1995JCB}).

\begin{theorem}\label{T_existssubwordkletters}
The number of $\exists$-subword queries needed to reconstruct an unknown word $w$ of unknown length $n$ over an alphabet $\mathcal{A}$ of cardinality $k$
is at most
$$n \lceil \log_2(k) \rceil +k \left(2+ \lfloor\log_2(n+1) \rfloor\right)\,.$$
\end{theorem}

Actually, our approach is similar to the method used in \cite{Skiena_Sundaram1993LNCS,Skiena_Sundaram1995JCB}.
We act essentially by dichotomy on the alphabet but when reconstructing words from their projections on a smaller alphabet we improve the bound on the number of queries. Also on small alphabets we use a linear decomposition instead of a binary decomposition in order to reduce the number of queries needed to deduce the number of occurrences of some letters.

To prove Theorem~\ref{T_existssubwordkletters} we use the next two lemmas.
The first one considers the reconstruction problem in the one letter alphabet case.
The second one describes upper bounds on the number of queries needed to reconstruct a word from projections on disjoint alphabets.

\begin{lemma}
\label{L_exists_one_letter}
Given an unknown nonempty word $w$ of length $n$ over an alphabet $\mathcal{A}$ and a letter $\alpha \in \mathcal{A}$,
the value $|w|_\alpha$ can be determined using
\begin{itemize}
\item at most $2\lfloor1 +\log_2(|w|_\alpha + 1)\rfloor$
$\exists$-subword queries if $n$ is unknown and
\item at most $\lceil \log_2(n+1) \rceil$
$\exists$-subword queries if $n$ is known.
\end{itemize}
\end{lemma}

The proof of this Lemma is a simple binary search. The details can be found in Appendix \ref{appproof}.
In the next Lemma we explain how to reconstruct a word $w$ from its projections on two disjoint
complementary alphabets. Note that \cite[Lemma 14]{Skiena_Sundaram1995JCB}, is almost the same result with a number of queries $2.18|\pi_{\mathcal{B}}(w)|+|\pi_{\mathcal{C}}(w)|+5$ instead of $|\pi_{\mathcal{B}}(w)|+|\pi_{\mathcal{C}}(w)|+1$. The main difference is that instead of using a binary search we simply go greedily from left to right when combining the two words. This lemma almost exclusively explains the improvement we obtain over \cite[Theorem 2]{Skiena_Sundaram1995JCB}.

\begin{lemma}\label{L_merge}
Let $w$ be an unknown word of length $n$ over an alphabet $\mathcal{A}$.
Let $\mathcal{B}$ and $\mathcal{C}$ be two disjoint alphabets such that $\mathcal{A}=\mathcal{B}\cup\mathcal{C}$, then
\begin{enumerate}
\item if we know both projections $\pi_{\mathcal{B}}(w)$ and $\pi_{\mathcal{C}}(w)$, then the word $w$ can be reconstructed using at most $n-1$ $\exists$-subword queries,

\item if we know the word $\pi_{\mathcal{B}}(w)$ and $\#\mathcal{C}=1$,
then the word $w$ can be reconstructed using at most $n+1$ $\exists$-subword queries.
\end{enumerate}
\end{lemma}

It may be observed that in item 1 of Lemma~\ref{L_merge},
the length of $w$ can be determined without asking any query since it is equal to
$|\pi_{\mathcal{B}}(w)|+|\pi_{\mathcal{C}}(w)|$.
This is not the case in item 2. In both cases, the length is not directly used in the proof.

For any word $x=x_1\cdots x_\ell\in \{0,1\}^\ell$  and integers $i,j\in\{1,\ldots, \ell\}$, let $x[i\ldots j]=x_ix_{i+1}\cdots x_j$ when $i\le j$. By extension, if $i> j$ (and possibly $i=|x|+1$ or $j=0$), then $x[i\ldots j]$ is the empty word.
\begin{proof}[Proof of Lemma~\ref{L_merge}]
Assume first that  $u=\pi_{\mathcal{B}}(w)$ and $v=\pi_{\mathcal{C}}(w)$ are known.
The first letter of $w$ is either $u_1$ or $v_1$. More precisely, $u_1v$ is a subword of $w$ if and only if $u_1$ is the first letter of $w$, otherwise $v_1$ is the first letter of $w$. Thus in one question we can determine the first letter of $w$, and the projections  $\pi_{\mathcal{B}}(w[2\ldots n])$ and $\pi_{\mathcal{C}}(w[2\ldots n])$. We can repeat this process and after each new query we obtain the next letter of $w$ and the two projections of the rest of $w$ over $\mathcal{B}$ and $\mathcal{C}$.
\begin{algorithm}
\caption{An algorithm that returns an unknown word $w$ over $\mathcal{B}\cup\mathcal{C}$ with $\mathcal{B} \cap \mathcal{C}=\emptyset$ from
$u = \pi_{\mathcal{B}}(w)$ and
$v= \pi_{\mathcal{C}}(w)$\label{A_merge}}
\begin{algorithmic}
\State $p \gets \varepsilon$ ; $i \gets 0$ ; $j \gets 0$
\While{$i < |u|$ \textbf{and} $j < |v|$}
    \If{$p u_{i+1} v[j+1..|v|]$ is a subword of $w$}
        \State $p \gets  pu_{i+1}$ ; $i \gets i+1$
    \Else
        \State $p \gets  pv_{j+1}$ ; $j \gets j+1$
    \EndIf
\EndWhile
\State  $p \gets p u[i+1..|u|] v[j+1..|v|]$
\State \Return  $p$
\end{algorithmic}
\end{algorithm}

Hence Algorithm~\ref{A_merge} allows to reconstruct $w$ from
$u$ and $v$. In this algorithm $i$ and $j$ store respectively the successive length of $\pi_{\mathcal{B}}(w[1\ldots i+j])$ and $\pi_{\mathcal{C}}(w[1\ldots i+j])$: at the beginning of each while loop, we know $p= w[1\ldots i+j]$.

From the preliminary comments, it is straightforward that at the end of the algorithm $p = w$ and that the number of $\exists$-subword queries asked is at most $n-1$.

From now on assume that we only know the word $\pi_{\mathcal{B}}(w)$ and the fact that $\mathcal{C}=\{a\}$ for some letter $a$.
We use a strategy similar to the previous case, that is,
we try to insert occurrences of the letter $a$ between the letters of $\pi_{\mathcal{B}}(w)$ in a greedy way.
Once the places of all letters of $\pi_{\mathcal{B}}(w)$ are known,
one has to determine the remaining occurrences of $a$ at the end of $w$.
This leads to the variant Algorithm~\ref{A_merge2}
for which the number of $\exists$-subword queries asked is exactly $n+1$: there is one query by letter of $\pi_{\mathcal{B}}(w)$ and $\pi_{\mathcal{C}}(w)$ and one additional query needed to determine when there is no more letter in $\pi_{\mathcal{C}}(w)$.
\begin{algorithm}
\caption{An algorithm that returns an unknown word $w$ over $\mathcal{B} \cup \{a\}$ with $a \not\in \mathcal{B}$ from $u = \pi_{\mathcal{B}}(w)$\label{A_merge2}}
\begin{algorithmic}
\State $p \gets \varepsilon$ ; $i \gets 0$
\While{$i < |u|$}
    \If{$p a u[i+1..|u|]$ is a subword of $w$}
        \State $p \gets  pa$
    \Else
        \State $p \gets  pu_{i+1}$ ; $i \gets i+1$
    \EndIf
\EndWhile
\While{$pa$ is a subword of $w$}
    \State  $p \gets pa$
\EndWhile
\State \Return  $p$
\end{algorithmic}
\end{algorithm}
\end{proof}

The proof of the next result explains the strategy to solve the reconstruction problem
using $\exists$-subword queries.
The length of $w$ may be unknown.

\begin{proposition}
\label{P_technical_exists_subwords}
Let $w$ be an unknown word over an alphabet of cardinality $k$.
For any $\mathcal{B}\subseteq\mathcal{A}$ with $\#\mathcal{B} \ge2$, the number of $\exists$-subword queries needed to reconstruct the word $\pi_{\mathcal{B}}(w)$ is at most
$$\lceil \log_2(\#\mathcal{B}) \rceil  |\pi_{\mathcal{B}}(w)|+\#\mathcal{B} \left(2+\max\limits_{\alpha \in \mathcal{B}} \lfloor\log_2(|w|_{\alpha}+1) \rfloor\right)\,.$$
\end{proposition}

\begin{proof}
We proceed by induction on the cardinality of $\mathcal{B}$ with the two base cases being $\#\mathcal{B}=2$ and $\#\mathcal{B}=3$.

If $\mathcal{B}=\{x,y\}\subseteq\mathcal{A}$ with $x\not=y$, we can apply Lemma \ref{L_exists_one_letter} to determine $\pi_{\{x\}}(w)= x^{|w|_x}$ in at most
$2 \lfloor 1+\log_2(|w|_x+1) \rfloor$ queries.
Case 2 of Lemma \ref{L_merge} implies that we can then determine $\pi_{\{x,y\}}(w)$ in at most $|\pi_{\{x,y\}}(w)|+1$ extra queries. The total number of queries is at most
$$|\pi_{\{x,y\}}(w)|+1+2 \lfloor 1+\log_2(|w|_x+1) \rfloor
\le\lceil \log_2(\#\mathcal{B}) \rceil  |\pi_{\mathcal{B}}(w)|+\#\mathcal{B} \left(2+\max\limits_{\alpha \in \mathcal{B}} \lfloor\log_2(|w|_{\alpha}+1) \rfloor\right)$$
as desired.

If  $\mathcal{B}=\{x,y,z\}$ for some distinct letters $x,y,z\in\mathcal{A}$, we use the strategy of the previous paragraph to determine $\pi_{\{x,y\}}(w)$ and we use case 2 of Lemma \ref{L_merge} once again to obtain $\pi_{\{x,y,z\}}(w)$ in at most $|\pi_{\{x,y,z\}}(w)|+1$ extra queries. The total number of queries is then at most
$$|\pi_{\{x,y\}}(w)|+|\pi_{\mathcal{B}}(w)|+2+2 \lfloor 1+\log_2(|w|_x+1) \rfloor
\le \lceil \log_2(\#\mathcal{B}) \rceil  |\pi_{\mathcal{B}}(w)|+\#\mathcal{B} \left(2+\max\limits_{\alpha \in \mathcal{B}} \lfloor\log_2(|w|_{\alpha}+1) \rfloor\right)$$
as desired.

We now have to deal with the induction. Assume $\#\mathcal{B}\ge4$. Let $\mathcal{C},\mathcal{C}'\subseteq\mathcal{B}$ be two disjoint alphabets such that $\mathcal{B}=\mathcal{C}\cup\mathcal{C}'$, $\#\mathcal{C}= \lfloor\frac{\#\mathcal{B}}{2}\rfloor$ and $\#\mathcal{C}'= \lceil\frac{\#\mathcal{B}}{2}\rceil$. The two last conditions imply
$$\lceil\log_2\#\mathcal{C}\rceil \leq \lceil\log_2\#\mathcal{C}'\rceil=\lceil\log_2\#\mathcal{B}\rceil-1\,.$$
By induction hypothesis, the number of queries to determine $\pi_{\mathcal{C}}(w)$ and $\pi_{\mathcal{C}'}(w)$ is at most
\begin{align*}
    &\lceil \log_2(\#\mathcal{C}) \rceil  |\pi_{\mathcal{C}}(w)|+\#\mathcal{C} \left(2+\max\limits_{\alpha \in \mathcal{C}} \lfloor\log_2(|w|_{\alpha}+1) \rfloor\right)\\
    +&
    \lceil \log_2(\#\mathcal{C}') \rceil  |\pi_{\mathcal{C}'}(w)|+\#\mathcal{C}' \left(2+\max\limits_{\alpha \in \mathcal{C}'} \lfloor\log_2(|w|_{\alpha}+1) \rfloor\right)\\
    &\le(\lceil \log_2(\#\mathcal{B}) \rceil-1)(|\pi_{\mathcal{C}}(w)|+|\pi_{\mathcal{C}'}(w)|)+(\#\mathcal{C}' +\#\mathcal{C})\left(2+\max\limits_{\alpha \in \mathcal{C}'\cup\mathcal{C} } \lfloor\log_2(|w|_{\alpha}+1) \rfloor\right)\\
    &\le(\lceil \log_2(\#\mathcal{B}) \rceil-1)(|\pi_{\mathcal{B}}(w)|)+\#\mathcal{B}\left(2+\max\limits_{\alpha \in \mathcal{B} } \lfloor\log_2(|w|_{\alpha}+1) \rfloor\right)\,.
\end{align*}
By case 1 of Lemma \ref{L_merge}, we  only need $|\pi_{\mathcal{B}}(w)|$ extra queries to determine $\pi_{\mathcal{B}}(w)$. In total, we used at most
$\lceil \log_2(\#\mathcal{B}) \rceil(|\pi_{\mathcal{B}}(w)|)+\#\mathcal{B}\left(2+\max\limits_{\alpha \in \mathcal{B} } \lfloor\log_2(|w|_{\alpha}+1) \rfloor\right)$ queries as required.
\end{proof}

\begin{proof}[Proof of Theorem~\ref{T_existssubwordkletters}]
Theorem~\ref{T_existssubwordkletters} is an immediate consequence of Proposition~\ref{P_technical_exists_subwords} taking $\mathcal{B}=\mathcal{A}$ and using $\max\limits_{\alpha \in \mathcal{B}} \lfloor\log_2(|w|_{\alpha}+1) \rfloor\le \lfloor\log_2(|w|+1)\rfloor$
\end{proof}

\section{Exists-factor queries\label{Sec:exists-factor_bis}}

In this section, we focus on queries in the form {\oq}Is $u$ a factor of $w$?{\cq}. Our aim is to prove Theorem 16. As for the result from [10] that we improve here, we assume in this section that the length of the word to determine is known.

A factor $u$ is said \textit{right-extendable} in a word $w$
if there exists a letter $a$ such that $ua$ is also a factor of $w$.
The word $ua$ is a \textit{right extension} of $u$.
A non-right-extendable factor $u$ of $w$ is a suffix of $w$
but the converse does not hold.
For instance the word $u = a$ is a suffix of the word $w =aa$
but it is right-extendable.
Actually it can be straightforwardly checked that
a factor $u$ is not right-extendable in $w$
if and only if $u$ is a suffix of $w$
which has only one occurrence as a factor of $w$.
The notions of left-extendability and  left  extensions are defined similarly.

The global strategy to reconstruct an unknown word $w$
using queries on factors is to apply the following three steps.
First we find a long  block of a fixed letter $\alpha$ (proof of Lemma~\ref{L_search_a^t_bis}).
Second we determine a non-right-extendable factor of $w$ having this long block of $\alpha$ as a prefix.
Two different approaches are developed in the proof of Lemmas~\ref{L_right_extendability_binary_bis} and \ref{L_right_extendability_binary_sqrt_bis}.
Depending on the length of the previously found long block of $\alpha$,  one or the other of the two approaches reveals to be more efficient.
Finally we determine $w$ from the previous non-right-extendable factor (Lemma~\ref{L_left_extendability_binary_bis}).
Let us first explain this last step.

\begin{lemma}
\label{L_left_extendability_binary_bis}
Let $w$ be an unknown word
of known length $n$ over an alphabet of cardinality $k$.
If we know a non-right-extendable factor $s$ of $w$
then we can reconstruct $w$ with at most $(k-1)(n-|s|)$  $\exists$-factor queries.
\end{lemma}

\begin{proof}
Assume that $|s|<n$. Then $s$ is a proper suffix of $w$.
Fix a letter $\alpha$.
We can ask {\oq}is $\beta s$ a factor of $w$?{\cq} for each letter $\beta$ different from $\alpha$.
 If the answer is positive for some $\beta$ then we know that $\beta s$ is a non-right-extendable factor of $w$
 and if the answer is negative for all $\beta$ then we know that $\alpha s$ is a non-right-extendable factor of $w$.
 We then repeat the same process until
 we reach a word of length $n$ (this word necessarily is $w$).
 It costs us at most $k-1$ queries by letter that we have to determine, that is, $(k-1)(n-|s|)$ queries.
\end{proof}

We now explain how to efficiently find a non-right-extendable factor of $w$.
For this a letter $\alpha$ is fixed and
we assume that we know the greatest $t$ such that $\alpha^t$ occurs as a factor in $w$.
And we will present two different strategies that we will use for different values of $t$ in the proof of Theorem \ref{T_exfac2_bis}.
The first strategy will be used when $t$ is not too large.
It is described in the proof of the following result.

\begin{lemma}
\label{L_right_extendability_binary_bis}
Let $w$ be an unknown word of known length $n$ over an alphabet ${\cal A}$ of cardinality $k$.  Let $\alpha \in {\cal A}$.
If we  know the largest integer $t$ such that $\alpha^t$ is a factor of $w$, then
a non-right-extendable factor $s$ of $w$ can be determined with at most $(k-1)(|s|+2)$  $\exists$-factor queries.
\end{lemma}

\begin{proof}
Let $\sigma$ be a variable that aims to contain the searched non-right-extendable factor of $w$.
We initialize $\sigma$ with the word $\alpha^t$.
We search for successive right extensions of $\sigma$
asking the query {\oq}is $\sigma\beta$ a factor of $w$?{\cq} for each letter $\beta\not=\alpha$.
If the answer is {\oq}yes{\cq} for some $\beta\not=\alpha$ then we know that $\sigma\beta$ is a factor of $w$
and we set $\sigma\beta$ to be the new value of $\sigma$.

If the answer is {\oq}no{\cq} for all $\beta\not=\alpha$, then either $\sigma\alpha$ is a factor of $w$ or $\sigma$ is non-right-extendable.
If $\sigma$ does not end with the suffix $\alpha^t$, we set $\sigma\alpha$ to be the new value of $\sigma$.
It is possible that $\sigma$ is no longer a factor of $w$ (and so $\sigma$ is not a non-right-extendable factor of $w$), in particular, when the previous value of $\sigma$ already was the searched non-right-extendable factor of $w$.
But if later, while trying to add a letter $\beta \neq \alpha$, we get {\oq}yes{\cq} as an answer
we deduce that we were right for every previous assumption.
If we obtain the answer {\oq}no{\cq} $t+1$ consecutive times
then we have added  $t+1$ occurrences of $\alpha$ at the end of $\sigma$.
This implies that we were wrong since by definition of $t$,
$\alpha^{t+1}$ is not a factor of $w$.
At this point $\sigma=v\alpha^{t+1}$ for some word $v$ that ends with a letter different from $\alpha$
and there exists $r\le t$
such that $v\alpha^{r}$ is a suffix of $w$
and both $v\alpha^{r+1}$ and all words $v\alpha^r\beta$ with $\beta \neq \alpha$ are not factors of $w$:
$v\alpha^r$ is the searched non-right-extendable factor of $w$.
We can determine $r$ by asking {\oq}is $v\alpha^{r+1}$ a factor of $w$?{\cq} from $r = 0$ and until a negative answer.

Let us now provide an upper-bound for the number of queries.
Let $v\alpha^{t+1}$ be the value of $\sigma$ obtained after $t+1$ consecutive negative queries
and let $r+1$ be the number of additional queries asked to determine the final value $s$ of $\sigma$.
Observe that $v$ was determined using $(k-1)(|v|-t)$ queries.
Then we use $(k-1)(t+1)$ queries to get $v\alpha^{t+1}$
and finally we use $r+1$ queries to determine the final value.
The total amount of queries is thus bounded by
$(k-1)((|v|-t)+(t+1)+(r+1))$.
Since $|s| = |v|+r$, this number of queries is bounded by $(k-1)(|s|+2)$.
\end{proof}

Let us illustrate in an example the strategy used in the proof of Lemma~\ref{L_right_extendability_binary_bis}.
Assume that the word to reconstruct is $w = 00011100111011$ and that we use $\alpha =1$.
We have $t = 3$ and initially $\sigma = 111$.
The answer to the two first queries are positive and we get $\sigma = 11100$.
Then the answers to the next three queries are negative and we assume that $\sigma = 11100111$ is a prefix of the expected result.
This is confirmed by the next query that sets $v = 111001110$. The next four negative queries on $v0$, $v10$, $v110$ and $v1110$ imply that
the  non-right-extendable factor is $v$, $v1$, $v11$, or $v111$.
After three additional queries, we know that $11100111011$ is a non-right-extendable factor (hence a suffix) of $w$.

If $t$ is large (essentially if $t \geq \lceil4 \sqrt{n}\,\rceil$; see the proof of Theorem~\ref{T_exfac2_bis}),
then a better strategy is to verify slightly more often that our assumptions are correct
when building the non-right-extendable factor. Doing so leads to the alternative strategy provided in the proof of the next result.
\begin{lemma}
\label{L_right_extendability_binary_sqrt_bis}
Let $w$ be an unknown word of known length $n$ over an alphabet ${\cal A}$ of cardinality $k$.  Let $\alpha \in {\cal A}$ be a letter with at least one occurrence in $w$.
Assume that we know $n$ and the largest positive integer $t$ such that $\alpha^t$ is a factor of $w$.
A non-right-extendable factor $s$ of $w$ can be determined using
at most $(k-1)(|s|-t)+k\left\lceil\sqrt{n}\,\right\rceil+1$  $\exists$-factor queries.
\end{lemma}
\begin{proof}
The strategy is almost identical to the previous one.
We initialize $\sigma$ with the word $\alpha^t$ and we try to extend it by asking whether $\sigma \beta$ for some $\beta \neq \alpha$ is a factor of $w$
and we proceed as previously.

If we obtain the answer {\oq}no{\cq} $r$ consecutive times
then we added $r$ occurrences of $\alpha$ at the end of $s$.
In this case, every $\left\lceil\sqrt{n}\,\right\rceil$ new consecutive occurrences of $\alpha$,
we verify if our current value of $\sigma$ is a factor of $w$.
If this holds we keep going.
Otherwise letting $v$ be the word such that
$\sigma = v\alpha^{\left\lceil\sqrt{n}\,\right\rceil}$,
$v\alpha^{\left\lceil\sqrt{n}\,\right\rceil}$ is not a factor of $w$.
We need to find the largest $r$ such that $v \alpha^r$ is a factor of $w$.
This can be done by setting $\sigma=v$ and asking the query {\oq}is $\sigma \alpha^i$ a factor of $w$ ?{\cq}, where $i$ starts at $1$ and increases until we receive the answer {\oq}no{\cq}.  

Let us now count the number of queries.
In the first phase, until reaching $v\alpha^{\left\lceil\sqrt{n}\,\right\rceil}$, the length of $\sigma$ increases from $t$ to $|v\alpha^{\left\lceil\sqrt{n}\,\right\rceil}|$. Each new letter requires at most $k-1$ queries, but each $\left\lceil\sqrt{n}\,\right\rceil$ query a verification query is done.
So the number of queries 
in this first phase is at most (remember $t \geq 1$)
$$(k-1)(|v\alpha^{\left\lceil\sqrt{n}\,\right\rceil}|-t)+ \left\lfloor\frac{|v\alpha^{\left\lceil\sqrt{n}\,\right\rceil}| -t}{\lceil\sqrt{n}\rceil}\right\rfloor\le
(k-1)(|v\alpha^{\left\lceil\sqrt{n}\,\right\rceil}|-t)+ 1+\left\lfloor\frac{|w|-1}{\lceil\sqrt{n}\rceil}\right\rfloor$$
which is upper-bounded by $(k-1)(|v|-t)+k\lceil\sqrt{n}\rceil\,$.

In the second phase there is one verification query and every other query increases the value of $i$ from $1$ to $r+1$. So there are at most $1+r=1+|s|-|v| \leq 1 + (k-1)(|s|-|v|)$ other queries in this second phase.
Summing the queries of the first and second phase, we deduce that at most $(k-1)(|s|-t)+k\left\lceil\sqrt{n}\,\right\rceil+1$ queries are used.
\end{proof}

Before using Lemma~\ref{L_right_extendability_binary_bis} or Lemma~\ref{L_right_extendability_binary_sqrt_bis}
we need to determine the greatest power of a letter in a word $w$.
This can be done using a binary search with queries in the form {\oq}Is $a^t$ a factor of $w$?{\cq} for $1 \leq t \leq n$.
A negative answer to the query {\oq}Is $a^1$ a factor of $w$?{\cq} shows that the letter $a$ does not occur in $w$. The next result holds for arbitrary alphabets.
Its proof specifies how the binary search is done.

\begin{lemma}
\label{L_search_a^t_bis}
Let $w$ be an unknown word.
Let $a$ be a letter, $x, y$ be two known integers and $t$ be the largest integer such that  $a^t$ is a factor of $w$.
If we know that $x\le t\le y$ then at most $\lceil\log_2(y+1-x)\rceil$ $\exists$-factor queries are needed to determine the value of $t$.
\end{lemma}
Once again the idea of this Lemma is to use a binary search and the details of the proof can be found in Appendix~\ref{annexproof2}. 

Applying successively Lemma~\ref{L_search_a^t_bis}, then Lemma~\ref{L_right_extendability_binary_bis} or Lemma~\ref{L_right_extendability_binary_sqrt_bis} and finally Lemma~\ref{L_left_extendability_binary_bis}, we get the next result.

\begin{theorem}\label{T_exfac2_bis}
An unknown nonempty word $w$ of known length $n$ over an alphabet of cardinality $k \geq 2$ can be reconstructed in at most $(k-1)(n+2) + \lceil\frac{\log_2n}{2}\rceil+3$  $\exists$-factor queries.
\end{theorem}
\begin{proof}
  We start with the query {\oq}is $\alpha^{\lceil4\sqrt{n}\,\rceil}$ a factor of $w$?{\cq}.  
 
 If we obtain a positive answer, we use Lemma~\ref{L_search_a^t_bis} (with $x=\lceil4\sqrt{n}\,\rceil$ and $y=n$ ($n \geq 1$))
  to compute the largest $t$ such that $\alpha^t$ is a factor of $w$ in at most $\lceil\log_2 n\rceil$ queries.  Then we apply Lemma~\ref{L_right_extendability_binary_sqrt_bis} to find a non-right-extendable factor $s$
in at most $(k-1)(|s|-t) + k\lceil\sqrt{n}\,\rceil+1$ queries.
Since $t \geq \lceil4\sqrt{n}\,\rceil \geq 4\lceil\sqrt{n}\,\rceil-3$,
$$(k-1)(|s|-t) + k\lceil\sqrt{n}\,\rceil+1 \le (k-1)(|s|+3)-(3k-4)\lceil\sqrt{n}\rceil+1\,.$$
We finally apply Lemma~\ref{L_left_extendability_binary_bis} to find $w$ in $(k-1)(n-|s|)$ queries. In this case, including the initial query, we need a total of at most $(k-1)(n+3) + \lceil\log_2n\rceil - (3k-4) \lceil\sqrt{n}\,\rceil+2\le (k-1)(n+2)$ queries (we use $k \geq 2$ and $n\ge1$ for this inequality).  

  If we obtain a negative answer, we use Lemma~\ref{L_search_a^t_bis} (with $x=0$ and $y=\lceil4\sqrt{n}\,\rceil-1$)
  to compute the largest $t$ such that $\alpha^t$ is a factor of $w$ in at most $\lceil\log_2(4\sqrt{n})\rceil=
  \lceil\frac{\log_2n}{2}\rceil+2$ queries.
  Then we apply Lemma~\ref{L_right_extendability_binary_bis} to find a non-right-extendable factor $s$ in $(k-1)(|s|+2)$ queries and we finally apply Lemma~\ref{L_left_extendability_binary_bis} to find $w$ in $(k-1)(n-|s|)$ queries. In this case we need a total of $(k-1)(n+2) + \lceil\frac{\log_2n}{2}\rceil+3$ queries including the initial query.
  \end{proof}
 
\section{Conclusion}

We have studied three reconstruction problems and, for each of them, we have improved upper bounds on the number of necessary queries.
For reconstruction of a word $w$ of length $n$ over an alphabet of cardinality $k$ using $\exists$-subword queries, we have a lower bound $n \log_2(k)$ and in Section~\ref{Sec:exist_queries}, we reduce the gap between the lower and the upper bound to an $O(k \log_2(n))$. An open question is whether this gap can be further reduced to an $O(k)$ number of queries or even lower.

For the reconstruction using $\#$-subword queries as considered in Section~\ref{Sec:How many}, up to our knowledge, no lower bound is known. Our upper bound is much lower than the previous one, but it could still be far from the truth. In particular, we showed that there exists a deterministic algorithm that requires in average $O(\log n)$ queries to reconstruct a uniform random binary word of length $n$, but this algorithm requires $\Theta(\sqrt{n\log n})$ queries in the worst case. This might be possible to find a deterministic algorithm that requires $O(\log n)$ queries in the worst case. We were not able to find a simple proof that this cannot be done in constant time only depending on the size of the alphabet.

For the reconstruction using $\exists$-factor queries as considered in Section~\ref{Sec:exists-factor_bis}, a simple counting argument yields the lower bound $n\log_2(k)$ on the number of queries. S.~S.~Skiena and G.~Sundaram provide in \cite{Skiena_Sundaram1995JCB} a lower bound in $kn/4 - o(n)$ queries which is better for large alphabets.
In the binary case, we were able to improve the gap between the lower and the upper bound, reducing it to $\left\lceil\frac{\log_2(n)}{2}\right\rceil+5$.
In the general case, even if our result improves the gap between the lower and upper bounds, this gap is still important.
As already mentioned in the introduction, the lower bound $kn/4 - o(n)$ given by S.~S.~Skiena and G.~Sundaram
is also valid if one considers queries in the form {\oq}What is the number of occurrences of $u$ as a factor of $w$?{\cq}. In some sense, considering the numbers of occurrences of factors does not bring a significant amount of extra-information for reconstruction comparatively to information on the existence of factors. 
This contrasts with the subword case where the number of occurrences gives much more information than the existence of occurrences.

To end, let us mention the existence, in the binary case, of a deterministic algorithm that requires, in average, $n+\mathcal{O}(1)$ $\exists$-factor queries over a uniform random word \cite{kzauo_et_al}
which is optimal up to an additive constant. The main idea of this algorithm is similar to the approach used in Section~\ref{Sec:exists-factor_bis}, but
the length $t$ of the longest block of $0$ is determined faster. Indeed, for a binary word of length $n$ taken uniformly at random, the average value of $|t-\log_2(n)|$ is in $\mathcal{O}(1)$.
The existence of a deterministic algorithm using an $n+\mathcal{O}(1)$ number of $\exists$-factor queries in the worst case is open.

\section*{Acknowledgment}
Authors thank Victor Poupet for useful discussions. Many thanks also for the referees for their accurate reading and their valuable suggestions.

\bibliographystyle{plain}
\bibliography{local}

\newpage

\appendix
\section{Proof of Lemma \ref{L_exists_one_letter}}\label{appproof}
\begin{proof}[Proof of Lemma \ref{L_exists_one_letter}]
Assume first that $n$ is unknown.
We start by finding $M$ the smallest power of $2$ larger than $|w|_\alpha$.
This can be done asking whether $\alpha^i$ is a subword of $w$ starting from $i = 1$ and doubling $i$ while the answer is positive. The upper bound is reached by $M=i$ when the answer is negative.

If $M = 1$, then $|w|_\alpha = 0$ and exactly one query was asked (and $1 \le 2\lfloor1 +\log_2(|w|_\alpha + 1)\rfloor$ as desired).
Otherwise, $M = 2^{\lfloor \log_2 |w|_\alpha \rfloor +1}$ is found in $\lfloor \log_2 |w|_\alpha \rfloor +2$ queries. In this case we know, $M/2\le |w|_\alpha< M$, and we can find the value of $|w|_\alpha$ by binary search. The interval $\{M/2,\ldots, M-1\}$ contains $2^{\lfloor \log_2 |w|_\alpha \rfloor}$ values, hence the binary search requires $\lfloor \log_2 |w|_\alpha \rfloor$ $\exists$-subword queries. In the whole process $|w|_\alpha$ can be found using $2 \lfloor 1+\log_2 |w|_\alpha \rfloor\le2\lfloor1 +\log_2(|w|_\alpha + 1)\rfloor$ $\exists$-subword queries as desired.

When $n$ is known, $n$ is an upper bound on $|w|_\alpha$ and
the binary search can be done in the interval $[0, n]$. Hence $|w|_\alpha$ can be determined using  at most $\lceil \log_2(n+1) \rceil$ $\exists$-subword queries.
\end{proof}

\section{Proof of Lemma \ref{L_search_a^t_bis}}\label{annexproof2}
\begin{proof}[Proof of Lemma \ref{L_search_a^t_bis}]
We proceed by induction on the value $y+1-x$. If $x=y$ then we know the value of $t$ and no more queries are needed as expected. If $y>x$, then we ask the query {\oq}is $a^{\lceil(x+y)/2\rceil}$ a factor of $w$?{\cq}.

We deduce $x'\le t\le y'$ where, if the answer is {\oq}yes{\cq}, $x'=\lceil(x+y)/2\rceil$ and $y'=y$ and, if the answer is {\oq}no{\cq}, $x'=x$ and $y'=\lceil(x+y)/2\rceil-1$. In the two cases,
\begin{equation}\label{xtox'_bis}
 y'-x'+1\le 1+\lfloor(y-x)/2\rfloor\,.
\end{equation}
The map $f:z\mapsto \lfloor\frac{z-1}{2}\rfloor +1$ is non-decreasing over the non-negative reals and for all integers $n$, $f(2^n)= 2^{n-1}$, thus for all $z\le2^n$, we have $f(z)\le 2^{n-1}$. Since \eqref{xtox'_bis} can be rewritten, $y'-x'+1\le f(y+1-x)$, we deduce that for all integers $n$, if $y+1-x\le2^n$ then $y'+1-x'\le 2^{n-1}$. In particular, choosing $n=\lceil\log_2(y+1-x)\rceil$ yields,
$y'+1-x'\le 2^{\lceil\log_2(y+1-x)\rceil-1}$, hence
\begin{equation*}
\lceil\log_2(y'+1-x')\rceil \le \lceil\log_2(y+1-x)\rceil-1\,.
\end{equation*}
By induction hypothesis, it implies that we need at most $\lceil\log_2(y+1-x)\rceil-1$ other queries to determine the value of $t$. With the initial query, this is a total of at most $\lceil\log_2(y+1-x)\rceil$ queries as desired.
\end{proof}
\end{document}